\documentclass[1p]{elsarticle}
\pdfoutput = 1
\usepackage{color}
\usepackage{amsmath}
\usepackage{amssymb}
\usepackage{amsthm}
\usepackage{graphicx}
\usepackage{subfigure}

\usepackage{cancel}
\newtheorem{lem}{Lemma}

\begin{document}

\begin{frontmatter}

\title{ Mechanisms of Multi-strain Coexistence in Host-Phage Systems with Nested Infection Networks}

\author[sof]{Luis F. Jover}
\ead{ljover3@gatech.edu}

\author[sob,som]{Michael H. Cortez}
\ead{michael.cortez@biology.gatech.edu}

\author[sof,sob]{Joshua S. Weitz\corref{cor1}}
\ead{jsweitz@gatech.edu}

\address[sof]{School of Physics, Georgia Institute of Technology, Atlanta, GA, USA.}
\address[sob]{School of Biology, Georgia Institute of Technology, Atlanta, GA, USA.}
\address[som]{School of Mathematics, Georgia Institute of Technology, Atlanta, GA, USA.}

\cortext[cor1]{Corresponding Author}

\begin{abstract}
Bacteria and their viruses (``bacteriophages") coexist in natural environments forming complex infection networks. Recent empirical findings suggest that phage-bacteria infection networks often possess a nested structure such that there is a hierarchical relationship among who can infect whom.  Here we consider how nested infection networks may affect phage and bacteria dynamics using a multi-type Lotka-Volterra framework with cross-infection.  Analysis of similar models have, in the past, assumed simpler interaction structures as a first step towards tractability.    We solve the proposed model, finding trade-off conditions on the life-history traits of both bacteria and viruses that allow coexistence in communities with nested infection networks.  First, we find that bacterial growth rate should decrease with increasing defense against infection.  Second, we find that the efficiency of viral infection should decrease with host range.   Next, we establish a relationship between relative densities and the curvature of life history trade-offs.  We compare and contrast the current findings to the ``Kill-the-Winner'' model of multi-species phage-bacteria communities. Finally, we discuss a suite of testable hypotheses stemming from the current model concerning relationships
between infection range, life history traits and coexistence in complex phage-bacteria communities.
\end{abstract}

\end{frontmatter}

\section{Introduction}
Bacteria and their viral parasites, i.e., phages, are found in natural environments from oceans, soils to the human gut.  There are an estimated $10^{30}$ bacteria on Earth, with estimates of phages approximately 10-fold higher \cite{Suttle2005,Suttle2007}. Phages are not only abundant,
but they are also key players in
ecosystems.  For example, phages can be responsible for a significant portion of microbial mortality, e.g., with estimates ranging from 20\%-80\% \cite{Suttle1994,Weinbauer2004,Corinaldesi2010}.  These estimates of lysis are at the community scale.  However, individual phages infect a subset of bacteria in a community.  A growing number of empirical studies have begun to investigate the nature of cross-infections between phages and bacteria  \cite{Wichels1998,Holmfeldt2007,Gomez2011,Poisot2011}.   These studies have the potential to help identify the basis for phage-induced mortality, by delineating the specific phage types capable of infecting and lysing specific host types and, potentially, the taxonomic and biogeographic drivers of cross-infection \cite{Weitz2013}.  Although predictive models of cross-infection remain elusive, it is evident that a single virus can infect multiple strains of a host \cite{Poullain2008,Stenholm2008}, in some cases multiple species \cite{Wichels1998,Holmfeldt2007}, and even hosts from different genera \cite{Sullivan2003}.

In an effort to identify cross-system trends and patterns in cross-infection, we (and collaborators) recently re-analyzed 38 phage-bacteria infection studies. This re-analysis identified a recurring pattern in these studies: viral strains have overlapping host ranges such that the overall network of infections is significantly nested  \cite{Flores2011}. In nested systems, there is a hierarchy for who can infect whom (see Figure 1A).  In a perfectly nested system, the specialist virus can infect the most permissive host, the next most specialized virus infects the most permissive host and the second most permissive host, and so on (see Figure 1B).  Hence, the host that is most difficult to infect is infected only by the most generalist virus. In systems that are nested but not perfectly so, this nesting of infection ranges occurs more frequently than expected by chance.  It is important to note that some of the nested phage-host infection networks re-analyzed in \cite{Flores2011} are derived from experimental evolution studies in which the hosts and phages do not coexist at the same time point, but rather the nested relationship is only observed when performing cross-infection experiments between time points (e.g. \cite{Poullain2008}).  However, some networks are derived from ecological studies where samples are taken from the environment (e.g., \cite{Wichels1998,Holmfeldt2007,Stenholm2008}).  In such cases the finding of overlapping infection ranges poses a dilemma for understanding coexistence. In particular,  how can a specialist virus coexist when it only infects a single host, and indeed the host that is most susceptible to infection?  Further, how can a permissive host coexist with other more resistant hosts when there is a nested relationship to infectivity?  The consequences of such interaction networks on the structure and dynamics of microbes and their viral parasites have not yet been established.  Here, we integrate empirical observations of complex infection networks into ecological models of host-phage dynamics.

Host-phage systems and their population dynamics have been studied mathematically for over 30 years \cite{Levin1977,Abedon2008}. The earliest models were meant to facilitate understanding of relatively low-diversity chemostat experiments, often involving the cross-infection of a single phage type with a single host type.  These studies often found that mutants could arise and so simple models were often extended to include two hosts and one phage, two phages and one host, and so on  \cite{Levin1977,Lenski1985}. However, host-phage models applied to natural environments require a greater diversity of bacteria and phage types.  There are different types of approaches to integrate diversity into dynamic models.  On the one hand, eco-evolutionary models have been developed that assume that new host and viral types can evolve as a result of ecological interactions.  Despite their computational complexity, multiple examples of such eco-evolutionary approaches for phage-host dynamics are now available \cite{Weitz2005,Forde2008,Weinberger2012,Childs2012}.  Alternatively, ecological models in which a fixed diversity of types are included from the outset can be used to study how certain features in a community may help maintain diversity.  Examples of such approaches include spatial multi-strain models \cite{Haerter2012}, models of competition between hosts possessing different types of immune systems \cite{Levin2010} and models that incorporate higher trophic levels \cite{Thingstad2000}.

The most prominent theory of phage-bacteria ecological dynamics in multi-species communities is of the latter type and is known as the Kill-the-Winner model \cite{Thingstad2000,Winter2010}.  In this model, multiple species of bacteria and phage are considered, for which each virus is assumed to exclusively infect a single host type.  The model's central conclusion is that each host type is controlled, in a top-down fashion, by a single viral type, with the exception of a single host type whose density is controlled by the total host biomass limit set by an additional generalist grazer.  Moreover, the steady state densities of viral  types are determined by relative differences in life history traits of the hosts (an issue we return to in the discussion).    The model is meant to describe the dynamics of bacterial species, i.e., no strain level dynamics are considered, although its structure is general.

Here we extend the basic framework of the Kill-the-Winner (KTW) model to incorporate complex interaction networks and ask how nestedness mediates coexistence in multi-strain host-phage systems. We focus our attention on the idealized case of a perfectly nested interaction network, and show  conditions on the life-history parameters necessary for coexistence. We find trade-off conditions necessary for coexistence and show that coexistence occurs even when the system is perturbed from the equilibrium. We examine the abundances of both hosts and viruses at steady state and their relationship to life history traits and to infection range. Finally we examine the existence of coexistence equilibria in the general case where the infection matrix is not perfectly nested.  We close by discussing the relevance of the current study to the KTW model, empirical efforts to link infection patterns with life history traits, and recent attempts to establish a link between network structure and biodiversity.

\section{Methods}

We model a system of $n$ bacterial strains (hosts) and $n$ viral strains (phages). Bacteria compete for implicit resources while viruses infect different subsets of the bacterial community. We denote the density of host $i$ by $H_i$ and the density of virus $j$ by $V_j$. Our model of the ecological dynamics of the different host and viral strains is,
\begin{equation}
\label{odeSystem}
\begin{aligned}
\frac{dH_i}{dt}&=  r_iH_i\left(1 - \frac{\sum\limits_{j=1}^nH_j}{K }\right) - \sum_{j=1}^nM_{ij}\phi_{j}H_iV_j,\\
\frac{dV_j}{dt}&= \sum_{i=1}^nM_{ij}\phi_{j}\beta_{j}H_iV_j - m_jV_j.
\end{aligned}
\end{equation}
In the absence of viruses, hosts exhibit logistic growth with exponential growth rate $r_i$ and a community-wide carrying capacity $K$. To simplify the model we assume that for a given host, intrastrain and interstrain competition are the same.
 The parameters $\phi_{j}$ and $\beta_{j}$ are the adsorption rate and burst size (virion release per infection)   of virus $j$ (which we assume is independent of host $i$). Virus $j$ decays outside the host at a rate $m_j$. $M$ denotes the infection matrix, where $M_{ij} = 1$ if virus $j$ can infect host $i$ and $M_{ij} = 0$ if virus $j$ cannot infect host i. For analytical tractability, there are several characteristics of the life cycle of a viral infections that are not included in this model. For example, we assume there is no delay between infection and virion release. We also do not include  the possibility of  lysogeny, where viral genetic material is  incorporated into the host chromosome and vertically transmitted to daughter cells for future activation and lysis.

\begin{figure}
\centering
\subfigure{\label{infnet_stenholm}\includegraphics{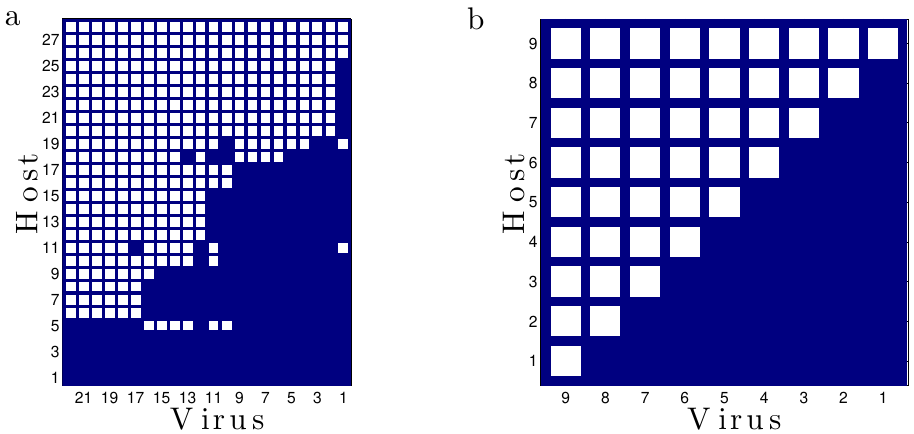}}
\caption{(a) Infection network from an experimental study presenting a statistically nested pattern (original data from \cite{Stenholm2008} reanalyzed in \cite{Flores2011}). The numbers identify different types of viruses and hosts. (b)  Perfectly nested infection network. For a perfectly nested network, the numbers correspond to the rank (i.e number of interactions). White squares denote that a given virus can infect the host.}
\label{infnets}
\end{figure}

Empirically measured phage-host infection networks are often statistically nested (for example, see Figure \ref{infnets}a). Here, our main goal is to examine the ecological implications of nested interaction networks for host-phage dynamics. We start by  studying the idealized case of a  system where the infection network is perfectly nested; see Figure \ref{infnets}b. In this  representation of the network, viruses are numbered from the most specialist to the most generalist. We will refer to this number as the rank of a species.  We see that virus 1 infects only one strain of host whereas virus $j$ infects $j$ different  host strains. Note that, in the nested case, virus $j+1$ can infect all strains of  hosts that virus $j$ can infect plus an additional one. Similarly, hosts are ordered by the number of viral strains that can infect them. Host $i$ can be infected by $i$ different  viral strains. Again, note that in the perfectly nested case, host $i$ can be infected by all viruses that can infect host $i+1$.
We relax these simplifying assumptions in section \ref{deviation},  where we consider the case of interaction networks that are not perfectly nested (as is the case in Figure \ref{infnets}a).

Numerical simulations of the ecological dynamics were done using  a Runge-–Kutta method (ode45 in MATLAB \cite{MATLAB2010}). The parameters (life-history traits) used in all the simulations are shown in table \ref{table:para} in  \ref{ap:para}. These parameters were chosen from a baseline of biologically realistic values found in the literature \cite{Wommack2000,Abedon2007}.

\section{Results}

\subsection{Equilibrium densities and conditions for coexistence}

We start by examining the equilibrium densities of the host and phage strains.
When the $2n$ equations of  system \eqref{odeSystem} are simultaneously zero, we find the equilibrium densities corresponding to strain coexistence:

\begin{equation}
H_n^*=h_1,\quad H_{n-1}^* = h_2-h_1,\quad \hdots,\quad H_1^*= h_{n}-h_{n-1}.
\label{hSteady}
\end{equation}

\begin{equation}
V_n^* = \frac{r_1}{\phi_n}\left(1 - \frac{h_n}{K }\right),\quad
V_{n-1}^* = \frac{(r_2 -r_1)}{\phi_{n-1}}\left(1 - \frac{h_n}{K }\right),\quad
\hdots\quad
V_{1}^* = \frac{(r_n -r_{n-1})}{\phi_1}\left(1 - \frac{h_n}{K }\right).
\label{vSteady}
\end{equation}

\noindent where the superscript $^*$ denotes equilibrium densities and  $h_j=\frac{m_j}{\phi_j\beta_j}$.
The equilibrium densities  of the hosts are expressed exclusively in terms of the life-history traits of the viruses ($h_j$). This is  due to the top-down control of  the hosts by the viruses. In Figure \ref{calc_hsteady} we show  graphically how to construct the host steady states. We see that  virus 1, which exclusively infects host $n$, determines the equilibrium density of that host. The value of this density is the ratio $h_1=\frac{m_1}{\phi_1\beta_1}$. This result is the same as in a system with just one host and one virus, and it can be interpreted as the host density necessary to support the virus infecting it (Figure \ref{calc_hsteady}a). Virus 2, on the other hand, requires a host density of $h_2=\frac{m_2}{\beta_2\phi_2}$ to persist in the system. This virus infects two hosts: one in common with virus 1 (host $n$), and an additional one (host $n-1$). However, the density of host $H_n$ is already set by virus 1, so the density of host $n-1$ is the difference between the required density for virus 2 ($h_2$)  and the density of virus 1 ($h_1$) (Figure \ref{calc_hsteady}b). Similarly, virus 3 infects one more type of host than virus 2, and thus, the density of the additional  host ($H_{n-2}$) is the difference between the density required by virus 3 to survive ($h_3$) and the density  set by virus 2 ($h_2$); see Figure \ref{calc_hsteady}c. The densities of the remaining hosts are determined in an analogous way, Eq \eqref{hSteady}. From this result we obtain that the total host biomass at steady state  is $\sum\limits_{j=1}^nH_j^* =\frac{m_n}{\beta_n\phi_n}$, i.e. it is set by the most generalist virus.

\begin{figure}[h]
\centering
\includegraphics{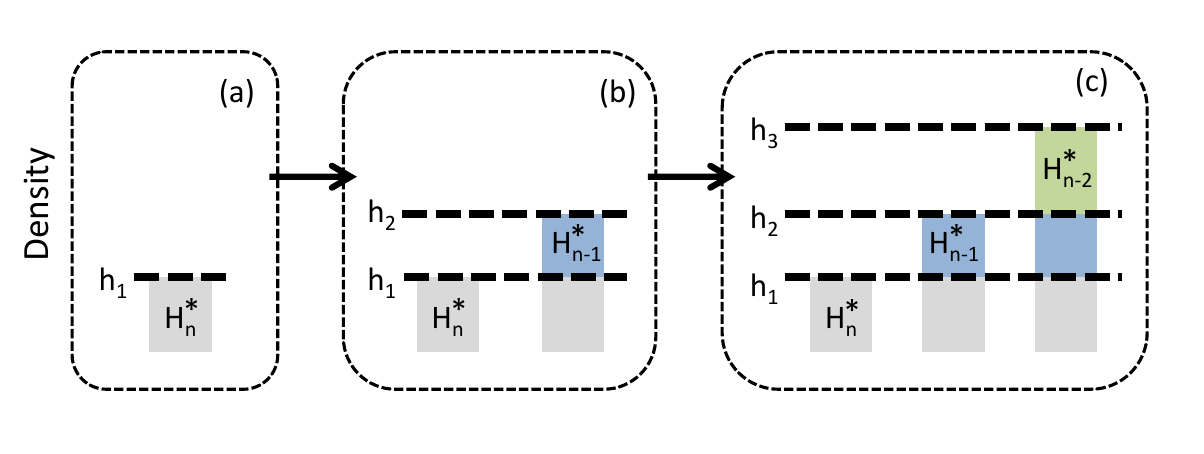}
\caption{Construction of the host steady states ($H_i^*$). (a) The equilibrium density of host $n$, $H^*_n$, is set by the aggregate life-history traits of virus 1, $h_1$, which infects $H_n$ exclusively. (b) The equilibrium density of host $n-1$, $H_{n-1}^*$, is set by the difference between $h_2$ and $h_1$. (c) The equilibrium density of host $n-2$, $H_{n-2}^*$, is set by the difference between $h_3$ and $h_2$.}
\label{calc_hsteady}
\end{figure}

Coexistence of all strains at positive densities requires three conditions to be met:

\begin{align}
&\quad\frac{m_{n}}{\beta_{n}\phi_{n}} > \frac{m_{n-1}}{\beta_{n-1}\phi_{n-1}}> \dots > \frac{m_2}{\beta_2\phi_2}  >\frac{m_1}{\beta_1\phi_1}\label{vTradeoff}\\[5pt]
&\quad r_n>r_{n-1}>\dots>r_2>r_1\label{hTradeoff}\\[5pt]
&\quad \sum\limits_{j=1}^nH_j^* = \frac{m_n}{\beta_n\phi_n}<K\label{Kcondition}.
\end{align}
These conditions establish a connection between the life-history traits of the different viral and bacterial strains and the  structure of the infection networks, specifically the rank of the strains.
Equation \eqref{vTradeoff} represents a tradeoff for the viruses and can also be written  as  $h_n>h_{n-1}>\dots>h_2>h_1$.  Larger values of $h_j$ imply a virus has a higher de-activation rate (higher $m_j$) and/or  produces fewer viruses per infection and/or is worse at attaching to the host (smaller $\beta_j\phi_j$).  Thus, the inequalities in equation \eqref{vTradeoff}  provide the  viral trade-offs necessary for coexistence: a virus with a broader host range has less advantageous life-history traits, as characterized by  its ratio $h_j$, compared to  viruses with a narrower host range. Note that if two viral types infected a single host, the type with the lower ratio $h_j$ would out-compete the other. Here, coexistence is possible because the viral type with higher $h_j$ can infect more host strains than the viral types with lower $h_j$. Equation \eqref{hTradeoff} describes the trade-off between immunity and growth rate for the hosts. Specifically, coexistence is possible if a host that can be infected by more viral types has a higher growth rate compared to a host that can be infected by fewer viral types. Finally,  equation \eqref{Kcondition} specifies that the sum of the host densities at equilibrium needs to be less than the carrying capacity of the system.

Figure \ref{2h2vStable} shows an example of the dynamics resulting from a system with 2 host and 2 viral strains which satisfies the conditions listed in equations \eqref{vTradeoff}, \eqref{hTradeoff}, and \eqref{Kcondition}. In contrast, Figure \ref{2h2vDie} shows what happens when the conditions for coexistence are not satisfied. In this example, host 1 has a larger growth rate than host 2, which results in the extinction of host 2, and as a consequence, the extinction of virus 1.

\subsection{Community dynamics and invasion }

Conditions \eqref{vTradeoff}, \eqref{hTradeoff}, and \eqref{Kcondition} guarantee the existence of a coexistence equilibrium in system \eqref{odeSystem}.  In this section we address two issues related to the coexistence between host and viral strains when those conditions are satisfied.  First, via numerical simulations we investigate if the species densities tend to the coexistence equilibrium or if cyclic coexistence is also possible.  Second, boundary equilibria where one or more of the host and viral strains are extinct also exist in system \eqref{odeSystem}.  We ask if those boundary equilibria are unstable with respect to invasion by the extinct host and viral strains when conditions \eqref{vTradeoff}, \eqref{hTradeoff}, and \eqref{Kcondition} are satisfied.

As seen in Figures \ref{2h2vStable}, \ref{2h2vOsc}, and \ref{5h5v}, when conditions \eqref{vTradeoff}, \eqref{hTradeoff}, and \eqref{Kcondition} are satisfied, the host and viral densities can either tend to steady state or exhibit cyclic oscillations.   In Figure \ref{2h2vStable}, the coexistence equilibrium point is stable, and the values of the densities tend to the equilibrium values after transient oscillations.  For a different set of life-history traits, the coexistence equilibrium point is no longer stable, but cyclic coexistence is still possible (Figure \ref{2h2vOsc}).  Figure \ref{5h5v} is an example of cyclic coexistence for 5 species of bacteria and 5 species of virus that satisfy the trade-off conditions \eqref{vTradeoff}, \eqref{hTradeoff}, and \eqref{Kcondition}.  In Figure \ref{5h5v}, coexistence is still possible because the coexistence equilibrium exists, even though the time series may be irregular.  Note that due to the dimension of the model we do not have a closed form solution for when cyclic dynamics arise in our system.

In appendix A we show that conditions \eqref{vTradeoff}, \eqref{hTradeoff}, and \eqref{Kcondition} imply that all boundary equilibrium points of system \eqref{odeSystem} are unstable with respect to invasion by at least one host or viral strain that is absent from that subsystem.  This implies that if the dynamics in that subsystem tend to the boundary equilibrium point, then that subsystem can be invaded by one or more of the extinct host or viral strains.  A stronger conclusion can be reached if system \eqref{odeSystem} is permanent (i.e. densities are bounded above and, after some time, are bounded below by a finite value \cite{Hofbauer1998}).  Because the average long term invasability conditions along orbits in permanent Lotka-Volterra systems are equal to the invasability conditions at equilibrium points \cite{Hofbauer1998}, conditions \eqref{vTradeoff}, \eqref{hTradeoff}, and \eqref{Kcondition} imply that the extinct strains can always invade the subsystem when system \eqref{odeSystem} is permanent.  It is an open question whether system \eqref{odeSystem} is permanent when conditions \eqref{vTradeoff}, \eqref{hTradeoff}, and \eqref{Kcondition} are satisfied.

\begin{figure}
\centering
\includegraphics{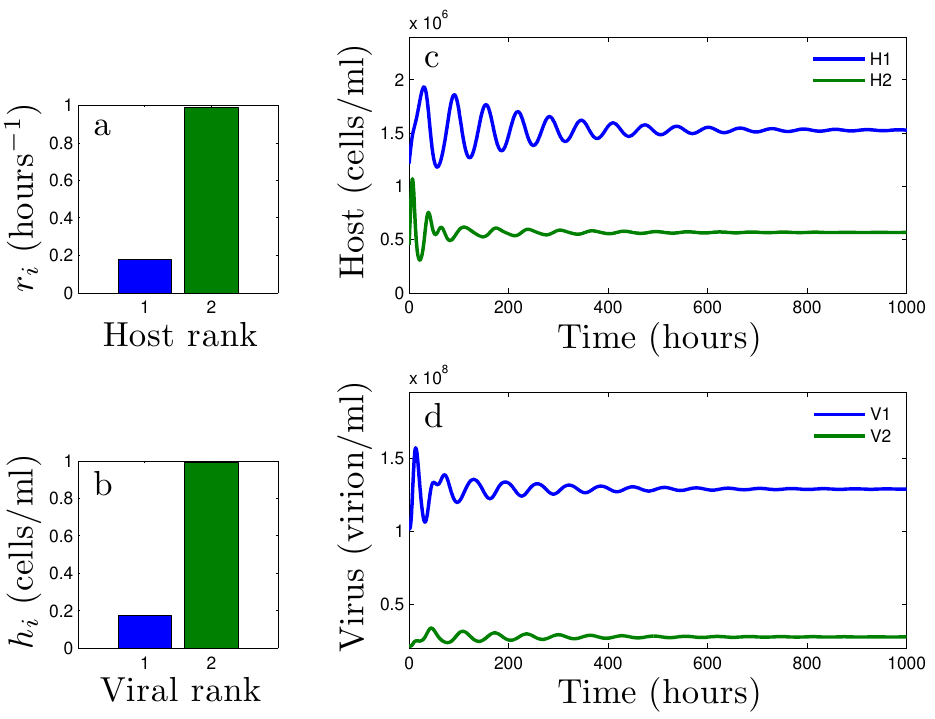}
\caption{Dynamics of 2 host types and 2 viral types when the infection network is perfectly nested and the  life-history traits satisfy the conditions \eqref{vTradeoff}, \eqref{hTradeoff}, and \eqref{Kcondition} for coexistence. (a) Growth rate, $r_i$, as a function of host rank in accordance with condition \eqref{hTradeoff}. (b) $h_i$ as a function of  viral rank in accordance with condition \eqref{vTradeoff}. (c) Host densities as a function of time: both hosts coexist via a stable equilibrium. (d) Viral densities as a function of time: both viruses coexist via a stable equilibrium.}
\label{2h2vStable}
\end{figure}

\begin{figure}
\centering
\includegraphics{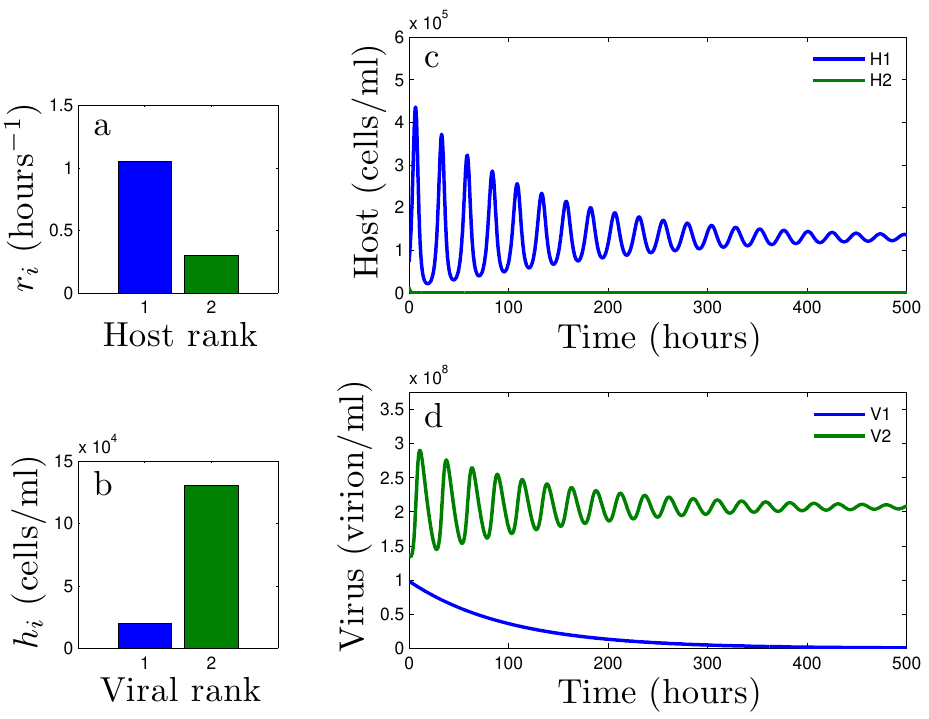}
\caption{Dynamics of 2 host types and 2 viral types when the infection network is perfectly nested and the trade-off conditions are not met.  One of the viral types does not survive because the life-history traits do not satisfy the conditions for coexistence. (a) Growth rate, $r$, as a function of host rank does not satisfy condition \eqref{hTradeoff}. (b) $h$ as a function of rank. (c) Host density as a function of time: host 2 goes extinct. (d) Viral densities as a function of time: virus 1 goes  extinct.}
\label{2h2vDie}
\end{figure}

\begin{figure}
\centering
\includegraphics{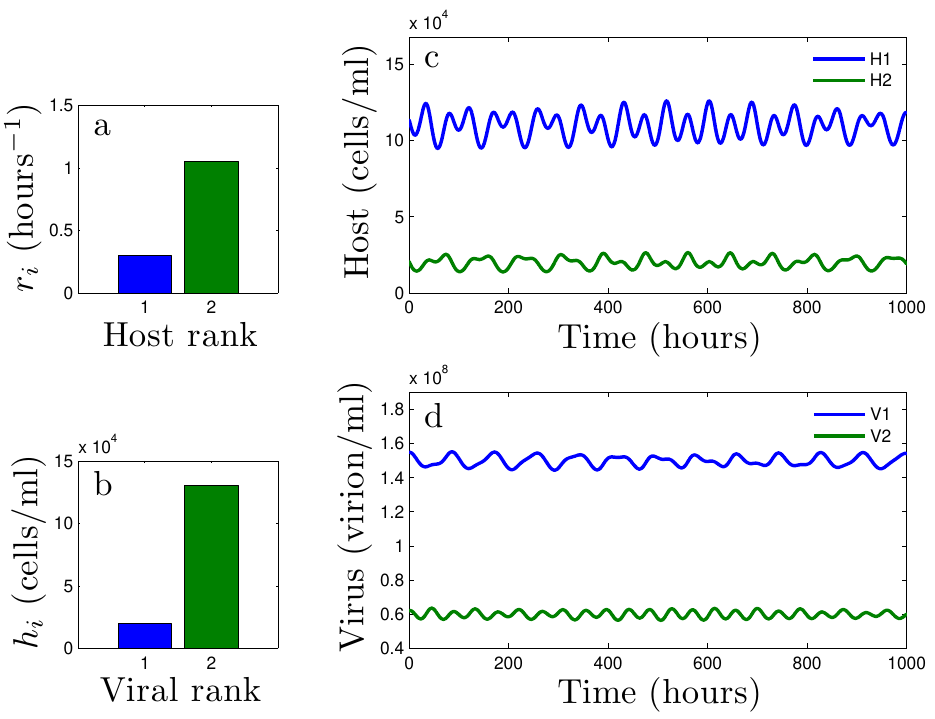}
\caption{Cyclic coexistence of 2 host types and 2 viral types.  The infection network is perfectly nested and the  life-history traits satisfy the conditions for coexistence. (a) Growth rate, $r_i$, as a function of host rank in accordance with condition  \eqref{hTradeoff}. (b) $h_i$ as a function of  viral rank in accordance with condition  \eqref{vTradeoff}. (c) Host densities as a function of time: both hosts coexist via a stable equilibrium. (d) Viral densities as a function of time: both viruses are present via oscillations.}
\label{2h2vOsc}
\end{figure}

\begin{figure}
\centering
\includegraphics{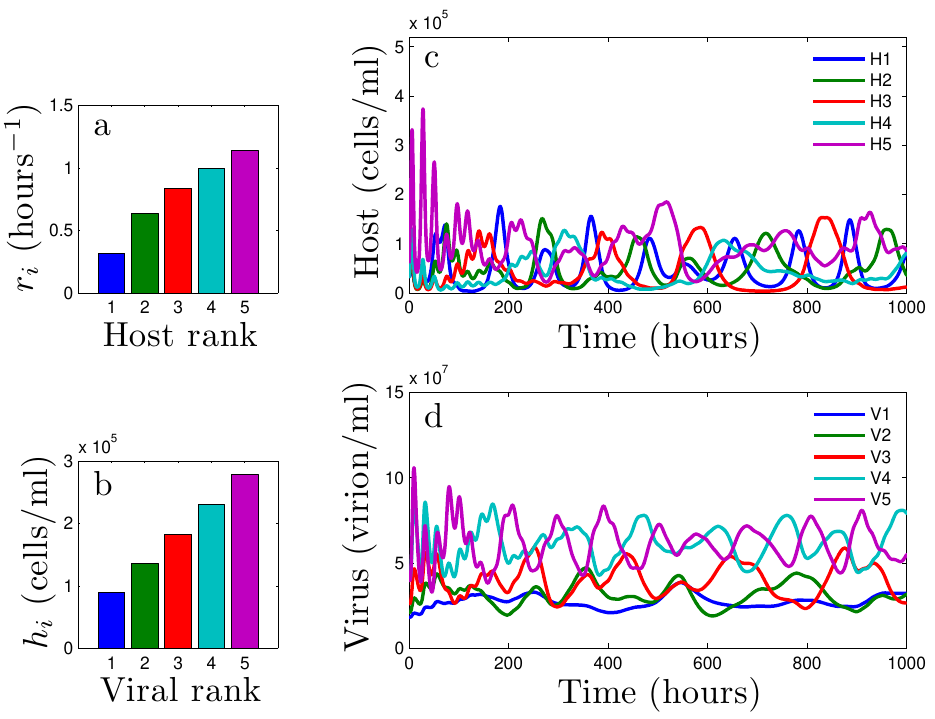}
\caption{Cyclic coexistence of 5 types of  hosts and 5 types of  viruses. The infection network is perfectly nested and the  life-history traits satisfy the conditions for coexistence. (a) Growth rate, $r$, as a function of host rank in accordance with condition \eqref{hTradeoff}. (b) $h$ as a function of  viral rank in accordance with condition \eqref{vTradeoff}, (c) host densities as a function of time: all of the  hosts coexist, (d) viral densities as a function of time: all of the viruses coexist.}
\label{5h5v}
\end{figure}

\subsection{Relationship between abundance and rank}

Another relevant question is the connection between  rank and density. Is there a way to infer  information about the infection network from measurements of density or vice-versa?
We find that, with the exception of host $n$, host densities are determined by the difference between consecutive $h_i$ values (eq. \eqref{hSteady}). So, in our framework, measurements of  density  could inform us about the differences in the life-history traits  of viruses with the most similar host range. If two viruses of consecutive rank are very similar (in terms of the aggregate life-history traits $h_i$), then the corresponding host density is low. On the other hand, two very different values of $h_i$ for consecutive viruses imply high host density. The only exception to this analysis is the density of host $n$, which is determined uniquely by the ratio $h_1$ of virus 1 and not by the difference between the ratios of viruses of consecutive rank.

Figure \ref{dens_rank} shows  examples of the connection between the life-history traits of the viruses and the density of the hosts. In the special case where the trade-off curve for the $h_i$ has a curvature with a constant sign, the values of $h_i$ translate into a simple (monotonic) rule for the host densities as a function of rank.  For example, when the trade-off between $h_i$ and viral rank is concave up, the  $H^*_i$ increase with rank (Figure \ref{dens_rank}a). When the trade-off is  concave down, the $H^*_i$ decrease  with rank (Figure \ref{dens_rank}b). In both cases, the density $H^*_n$ of host $n$  can be an exception.
Figure \ref{dens_rank}c  shows the general case. The values of  $h_i$ increase with rank following condition \ref{vTradeoff} for coexistence, but the corresponding equilibrium densities need not be a monotonic function of rank.

\begin{figure}
\centering
\includegraphics{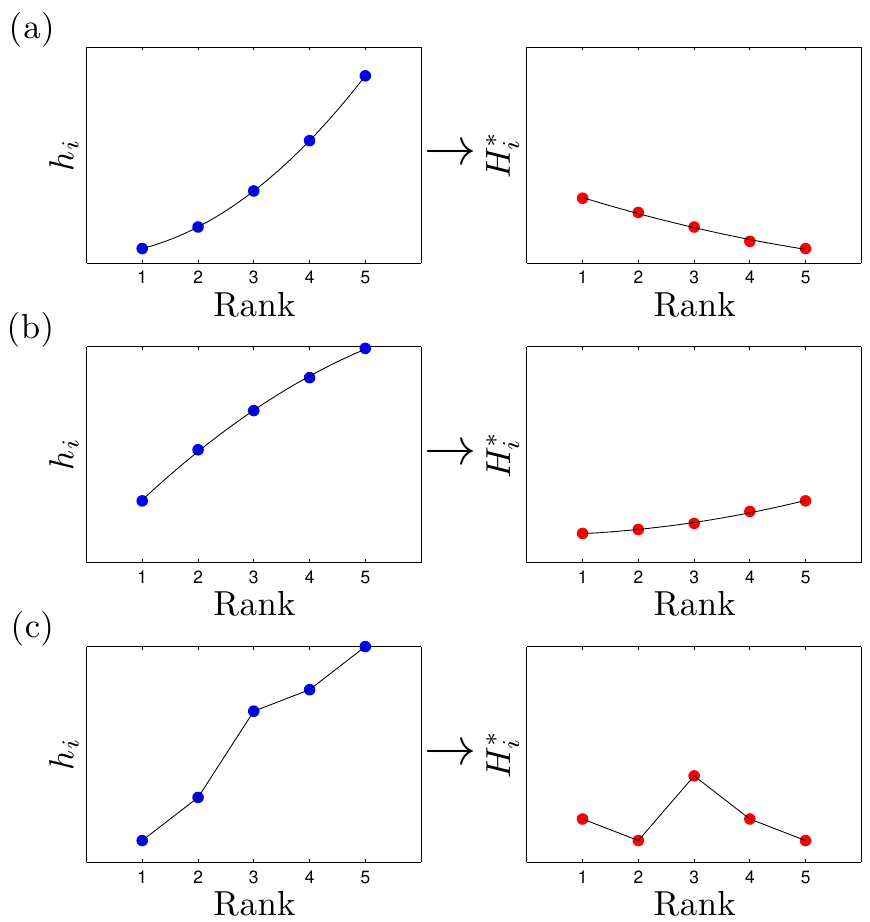}
\caption{Relationship between the viral  tradeoff curve (left) and the corresponding host densities at steady state (right). (a) When the trade-off curve is concave up, host density decreases with rank. (b) When the trade-off curve is concave down, host density increases with rank. (c) In general, the host densities are set by the difference of consecutive $h_i$ and need not  be a  monotonic function of rank.}
\label{dens_rank}
\end{figure}

 The viral densities depend, in part, on how different the host strains are. Specifically, the equilibrium densities of the viruses depend on the differences between growth rates of host strains of consecutive rank. However, they don't depend uniquely on the traits of the host, they also depend on the adsorption rate of the focal virus (eq. \eqref{vSteady}). Thus, information about the viral densities need not translate directly into  information concerning differences in the life-history traits of hosts.

\subsection{Deviation from  a perfectly nested network and coexistence}\label{deviation}

We now return to our assumption about interaction networks and consider systems that are not perfectly nested. The examination of imperfect nestedness is biologically relevant.  For example, a re-analysis of 38 observational studies of phage-host interaction networks found numerous instances of elevated nestedness, all of which included some departures from perfect nestedness \cite{Flores2011}.  The ecological dynamics follow system \eqref{odeSystem}, however the different structures of the infection network will be reflected in the matrix $M$, which includes information of who can infect whom.

The equilibrium densities can be expressed in  a compact form as two matrix equations using  the infection matrix $M$ and its transpose $M^\mathsf{T}$:

\begin{equation}
M^\mathsf{T}\vec{H}^* = \vec{h},\quad \quad
M\vec{V}^{*\prime} = \vec{r}\left(1 - \frac{\sum\limits_{j=1}^nH_j^*}{K }\right)
\label{steadyMatrix}
\end{equation}

Here $\vec{H^*}$ and $\vec{V^*}$ are vectors of the equilibrium densities, $\vec{h}$ is a vector whose elements are the  $h_i$, and $\vec{r}$ is a vector whose elements are the $r_i$. We also use the change of variable $V_i^{* \prime}=\phi_i V_i^*$. We are interested in the solutions of system \eqref{steadyMatrix} that are positive. We consider two cases: the infection matrix is invertible and the infection matrix is singular. An invertible infection matrix can be interpreted biologically as each viral strain having an unique niche, where the niche is defined by the host range. A singular infection network can be interpreted as the existence of niche overlap  between two or more viral strains.

\begin{figure}
\centering
\includegraphics{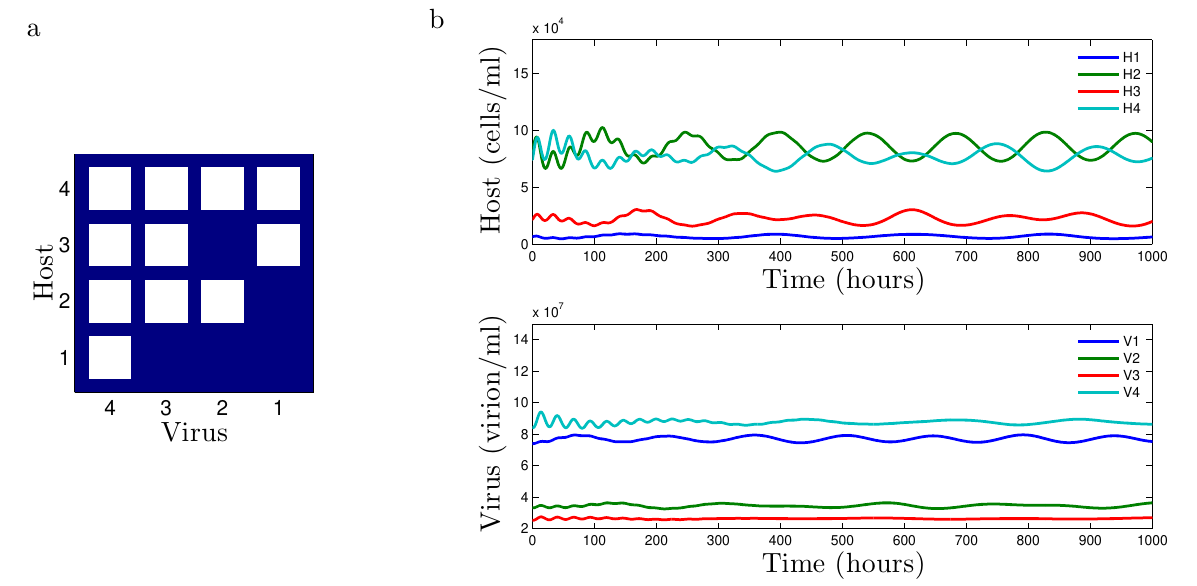}
\caption{ (a) Infection networks corresponding to an invertible infection matrix. (b) Time series from numerical simulation using the infection matrix in (a).}
\label{4by4}
\end{figure}

First consider the case where $M$ is invertible, i.e., there is niche differentiation. In this case, there exists a unique solution to system \eqref{steadyMatrix}, and therefore a unique set of coexistence equilibrium densities. The equilibrium densities  are expressed in terms of the life-history traits of the different viruses and bacteria and will, in general, involve differences between traits or combinations of traits (e.g. $h_i$). Therefore, for every invertible infection network there exists a series of inequalities that are necessary to guarantee positive equilibrium densities (analogous to conditions \eqref{vTradeoff}, \eqref{hTradeoff}, and  \eqref{Kcondition} for the perfectly nested case). In Figure \ref{4by4}a we show an example of an infection network that corresponds to an  invertible infection matrix that is not perfectly nested. We also show time series for a system with that specific  interaction matrix   where all the strains are present in the community (Figure \ref{4by4}b). The relationship between network structure and equilibrium densities may be suitable for further numerical analysis.

If the matrix $M$ is not invertible, i.e. there is niche overlap between viral strains, then there are two ways in which coexistence can occur.  The first case occurs when the life-history traits of some viral strains are effectively equal.  For example, in a system with two viral strains and one host strain, the viral strains can coexist if they have the same $h_i$ values.  Biologically, this case is unlikely to occur given that we are describing strains in terms of function, and in this case, coexistence is possible only when strains are functionally identical.

The second case in which coexistence is possible  occurs if we relax assumptions about the adsorption rate and the burst size for viral strains. In our formulation, a specific viral type infects all host strains with the same adsorption rate and burst size. However, if viruses exploit host strains at different rates, i.e., adsorption rate and burst size depend on both viral and host type, then coexistence is possible even with complete niche overlap. In this case, the steady states are solutions of two analogous systems of equations where the infection matrix $M$ is a weighted matrix whose entries correspond to the adsorption rate and burst size for each host-viral interaction.  A limited number of studies have shown that viral infection rates can differ significantly between host strains (e.g., see \cite{Flores2011}).  Hence, biologically, the use of quantitative information for host-phage infection assays represents an important target for future analysis and may shed light on the drivers of coexistence in natural communities.

\section{Discussion}

We studied the ecological dynamics of phages and hosts using a Lotka-Volterra framework that incorporated complex cross-infection networks. We found that  coexistence is possible even when viruses exploit overlapping ranges of host. In the case of a perfectly nested infection network, Figure \ref{infnets}b, we found  trade-offs for both the hosts and the viral strains that are  necessary to allow the coexistence of all the species.  The trade-off for hosts implies that host growth rate decreases with defense (the larger the growth rate, the larger the viral range). The trade-off for viruses implies that viruses must be less efficient at utilizing host resources as they increase their host range. We also showed that the densities of the host strains at equilibrium are determined by how different the viruses are in terms of their aggregate life-history traits.

The idealized case of a perfectly nested network is  mathematically tractable and  may help to identify potential principles underlying coexistence in host-phage systems. On the other hand, real infection networks are rarely perfectly nested (see Figure 1 for an example, and the re-analysis of Flores et al \cite{Flores2011}). We showed that coexistence is possible in the current framework if  there exists a partitioning of niches. In those cases, the resulting tradeoffs between infection range and life history traits are not easily presented in a general way because each  infection network results in different  trade-off conditions. This, in turn, makes it difficult to study the stability for the general case.  Nevertheless,  we showed numerical  examples of coexistence  away from equilibrium for a infection matrix that is not perfectly nested.

The current model is similar in spirit to the Kill-the-Winner (KTW) model \cite{Thingstad2000}.  The KTW model also proposed mechanisms by which bacteria could coexist, stabilized by the presence of viruses. However, in the KTW model, infections are one-to-one, meaning that each bacterial strain  can only be infected by one virus and likewise, each virus can only infect one bacteria. This would be equivalent to a diagonal infection network in our representation. In the KTW model, coexistence of different types of bacteria is achieved  through what is known as a ``killing the winner" mechanism, where ``coexistence among bacteria is ensured by host-specific viruses that prevent the best bacterial competitors from building up" \cite{Thingstad2000}.  Hence, viruses enable coexistence in a system that would, in their absence, lead to a diversity collapse.  The principle of top-down control also applies to the model presented here, despite the fact that we considered complex interaction networks with the possibility of  overlapping host range.

However, the current model makes predictions not found in the KTW model. First, KTW assumes each
bacterial strain is infected by a single viral type and so an ordering of bacterial growth rates
is not necessary for coexistence. In the current model, an ordering of bacterial growth rates with respect to susceptibility to infection is required for coexistence and represents a prediction of the model. Second, whereas KTW assumes no difference in viral life history traits, we again predict that differences in life
history traits are required for coexistence when viruses differ in their host range. Hence,
altogether we predict that there should be an entanglement between network structure and
bacterial/viral life history traits.

In light of the differences in predictions between our
models and the KTW model, we suggest that  studies of cross-infection from the enviornment should move
beyond qualitative analysis (i.e., whether or not a phage infects a bacteria) to quantitative
analysis (i.e., the lysis rate of bacteria by phages). Doing so would help identify costs of
resistance and infectivity in natural populations as well as help understand the relative
importance of network structure and trade-offs in shaping the structure of interactions in
natural communities.  We note that extensions of the current model will be required to consider the
type of trade-offs that may coincide with findings of imperfectly nestedness infection networks.

In the past, the relationship between infection range and physiological costs has been studied experimentally, often in model organisms. For example, Lenski studied 20 \emph{E. coli} mutant strains resistant to phage T4 \cite{Lenski1988a}. All 20 strains were significantly less fit that the parental  strain as measured via direct competition. Similarly, Bohannan and colleagues found strong growth costs to resistance for \emph{E. coli} strains resistant to phage $\lambda$ and to phage T4 \cite{Bohannan2000}. The hypothesized mechanism is that resistance often requires modification to surface receptors that may alter the rate and effectiveness of nutrient uptake \cite{Breitbart2012}.

However, resistance mechanisms need not depend on surface changes and, moreover, in practice such trade-offs are not found universally. Indeed, two recent studies of putative trade-offs found more equivocal evidence. In one instance, Lennon et al \cite{Lennon2007} studied Synechococcus hosts and associated myoviruses. They considered 22 different bacterial strains selected for viral resistance, based on four different ancestral strains. 11 out of the 22 strains showed a fitness cost compared to its ancestral strain either in  the form of its maximum growth rate from growth curve data or through direct competition assays. Hence, 11 of the 22 strains showed no fitness cost to resistance.  Similarly, Avrani et al \cite{Avrani2011} studied the relationship between Prochlorococcus hosts and associated podoviruses. They found that 11 of 23 mutant strains that had evolved phage resistance grew significantly more slowly than did the ancestral wildtype.  Hence, 12 of the 23 mutants strains did not show a growth cost to resistance.   Instead, they found an alternative type of trade-off in which resistant mutants were  susceptible to infection and lysis at a higher rate by other podoviruses and myoviruses with which they had not coevolved \cite{Avrani2011}. Hence, it may be that the costs of resistance can not be fully understood without taking the community into account, as was considered in the model presented here.

There have also been a limited number of studies linking host range expansion to fitness cost within phages.  For example,  Duffy and colleagues showed that the fitness of evolved RNA-based phages, when grown on the ancestral host, declined significantly in 28 of 30 phages that had evolved an expanded host range \cite{Duffy2006}.  However, 2 of the phages had no significant fitness cost \cite{Duffy2006}.  Given the diversity of hosts and viruses in natural systems, it seems important to extend prior assays relating one-step range expansions to fitness.  In particular, our work and the above suggest it is important to measure both host range and physiological costs in relationship to a broad range of potentially interacting strains.

The integration of life history trait variation into studies of community composition requires taking into account the coevolution of phages and bacteria.   Previous studies have explored the coevolution of phage-host
theoretically \cite{Weitz2005,Roswall2006,Forde2008,Rodriguez2009,Weinberger2012,Childs2012}. Experimental studies of co-evolving phage-host
systems suggest that infection networks are not static, but rather are dynamic and
reflect the changing identities of strains in a population  \cite{Buckling2002,Stern2011,Buckling2012}.  Phage-bacteria
coevolution is often described in terms of an arms race where the species' traits
escalate over time \cite{Stern2011,Breitbart2012}.  In such models it is common for better defended host
types and phage types with better offenses to replace the respective resident type.  However,
coexistence of types is also possible if disruptive selection or evolutionary branching
occurs.  We have shown that ecological coexistence between multiple phage and bacterial types is
possible if a series of trade-offs are satisfied (see Equations \eqref{vTradeoff}-\eqref{Kcondition}).  Thus, if there is an appropriate trade-off
between host defense and growth rate (e.g., resistance to infection) and an appropriate
trade-off between phage replication and host range (e.g., ability to bypass host
defense), evolutionary branching could yield phage-bacteria communities with nested
patterns of infection.  Thus, an arms race between phage and bacteria could result in
the coexistence of one phage and bacteria pair that evolves over time or the coexistence
of multiple phage and bacteria types that lie along their respective trade-offs.

We also point out that coevolution also occurs over spatially extended domains.  Spatial
structure is thought to stabilize diverse interactions amongst phage and bacteria (e.g.
\cite{Kerr2006,Abedon2007,Haerter2012}). Moreover, geographic structure can play a key role in affecting the
outcome of coevolution \cite{thompson2005geographic}. Further work is warranted to quantify how structure in
infection networks are driven by and act as drivers of the spatial distributions of
diverse communities of phages and bacteria \cite{Flores2012,green_2006,angly2006}.

Altogether, the current study predicts that phage-host coexistence in a given system depends on both network
structure (e.g., nestedness) and life history traits.  Other types of infection networks
are possible including modular \cite{Flores2011} and multi-scale \cite{Flores2012}, suggesting the need
for further investigations into the relationship between cross-infection and life history traits.  Extending the current model framework to include evolution, spatial dynamics, and quantitative information on infection will yield increased opportunities for theories of phage-host interactions to deepen our understanding of the mechanisms underlying the formation of complex phage-host interactions networks and their ecosystem-level consequences.

\section{Acknowledgments}
The authors thank two anonymous reviewers for their comments and suggestions.  JSW holds a Career Award at the Scientiﬁc Interface from the Burroughs Wellcome Fund and acknowledges the support of a grant from the James S. McDonnell Foundation and NSF No. OCE-1233760. MHC was partially supported by the National Science Foundation under Award No. DMS-1204401.

\appendix
\setcounter{figure}{0}
\renewcommand\thefigure{\Alph{section}.\arabic{figure}}

\section{Stability of the boundary fixed points for phage-host dynamics with a perfectly nested infection matrix}
\label{ap:stability}
In this section we show that if conditions \eqref{vTradeoff}, \eqref{hTradeoff}, and \eqref{Kcondition} from the main text are satisfied, then all boundary equilibria of the system are unstable. As discussed in the main text, this implies that if the dynamics in the boundary subsystem tend to the boundary equilibrium, then the subsystem can be invaded by at least one of the extinct host or viral strains. If trajectories in the subsystem do not converge to the boundary equilibrium point, then our conclusions about invasability only apply if system \eqref{odeSystem} is permanent (i.e. densities are bounded above and, after some time, are bounded below by a finite value, \cite{Hofbauer1998}). When system \eqref{odeSystem} is permanent, invasion at the equilibrium implies invasion along any orbit (over infinite time).  This is because the average long term invasability conditions along orbits in permanent Lotka-Volterra systems are equal to the invasability conditions at equilibrium points \cite{Hofbauer1998}. We are not aware of any proof that our system is or is not permanent.

\begin{figure}
\centering
\includegraphics{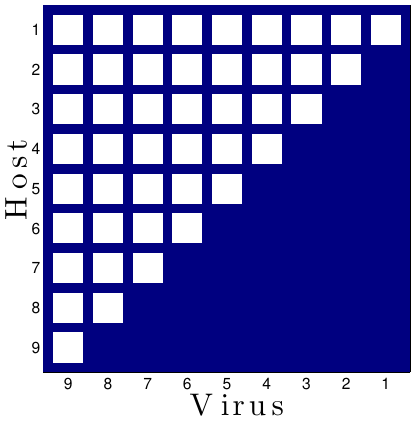}
\caption{ Perfectly nested infection network with the new host numbering used in \ref{ap:stability}.}
\label{infnet_nesA}
\end{figure}

For notational convenience,  we  make two changes to the way  the system is presented. First, we don't use the matrix $M_{ij}$ to establish who can infect who.  Instead we change the range of sums in accord with the nested pattern. Second, we  reverse the order of the hosts ($H_i \rightarrow H_{n-i +1}$). In this numbering, host 1 can be infected by all viral strains, and host $n$ can be infected by a single viral strain (Figure \ref{infnet_nesA}). In this new notation the system takes the form:

\begin{align}
\frac{dH_i}{dt} &=  r_iH_i\left(1 - \frac{\sum\limits_{j=1}^nH_j}{K }\right)  - H_i\sum_{j=i}^n\phi_jV_j, \label{eq:odeSumsH}\\
\frac{dV_i}{dt} &=  \phi_i\beta_iV_i\left(\sum_{j =1}^{i}H_j\right) - m_iV_i.\label{eq:odeSumsV}
\end{align}

In this  notation, conditions \eqref{vTradeoff}, \eqref{hTradeoff}, and \eqref{Kcondition} from the main text become:

\begin{align}
&\quad\frac{m_{i+1}}{\phi_{i+1}\beta_{i+1}}>\frac{m_i}{\phi_i\beta_i}\label{cond_h}\\
&\quad r_{i+1}<r_i\label{cond_r}\\
&\quad\frac{m_n}{\phi_n\beta_n}<K\label{cond_K}.
\end{align}
Note that, as a consequence of reversing the order of the hosts, condition \eqref{cond_r} on the growth rates ($r_i$) is the reverse of condition \eqref{hTradeoff}.

We will show that all possible boundary fixed points are unstable with respect to invasion. Let  $\textbf{X}= \left(H_1,H_2,\dots,H_n,V_1,V_2,\dots,V_n\right)$  denote the  vector of all host and viral densities and let   $\textbf{X}^*=\left(H_1^*,H_2^*,\dots,H_n^*,V_1^*,V_2^*,\dots,V_n^*\right)$  denote a  fixed point of system \eqref{eq:odeSumsH} and \eqref{eq:odeSumsV}, where $H^*_i$ is the equilibrium density of host $i$ and $V^*_i$ is the equilibrium density of virus $i$. We will  denote the Jacobian by $J$ and the eigenvalues by $\lambda$. We will consider 4  properties of fixed points, such that every possible boundary fixed point has at least one (and possibly more) of these properties. We will show that any fixed points with any of these properties is unstable with respect to invasion by at least one of the host or viral strains. The 4 properties of fixed points are:

\begin{enumerate}
\item  $\textbf{X}^*$ where $H_i^*=0\quad\forall\quad i\in [1,n]$.
\item  $\textbf{X}^*$  where $H_1^*=H_2^*=H_3^*=\dots=H_k^*=0$ with $k<n$ and $H_{k+1}^*\ne0$.
\item  $\textbf{X}^*$  where $H_i^* = H_{i+1}^*=\dots=H_n^* = 0$ for $i\neq1$ and $H_{i-1}\neq0$.
\item  $\textbf{X}^*$  where $H_i^* = H_{i+1}^*=\dots=H_k^* = 0$ for $i>1$, $k<n$, and $H_{k+1}\neq0$.
\end{enumerate}

The first property corresponds to fixed points where none of the hosts are present. The second property corresponds  to fixed points where hosts with low defense level are not present (i.e. hosts that can be infected by all or many viral types are not present). The third one corresponds to fixed points where hosts of high defense are not present (i.e. hosts that can be infected by one or only a few viral types are not present). The last property corresponds to fixed points where host of intermediate defense level are not present.

We will prove the instability of fixed points with properties 1,2,3, and 4 in lemmas 3, 4,5 and 6 respectively. Before that we will prove two more lemmas.

\begin{lem}
\label{l:eigenvalue}
 Let $\textbf{X}$ be a fixed point, then:
\begin{enumerate}[a)]
  \item  If $H_i^* = 0$, then  $\frac{\partial \dot{H_i}}{\partial H_i}\big|_{\textbf{X}^*}$ is an eigenvalue of the Jacobian evaluated at $\textbf{X}^*$.
  \item  If $V_i^*=0$, then $\frac{\partial \dot{V_i}}{\partial V_i}\big|_{\textbf{X}^*}$ is an eigenvalue of the Jacobian evaluated at $\textbf{X}^*$.
\end{enumerate}
\end{lem}

\begin{proof} a)

Assume $H_i^*=0$. We can  write the equation describing the dynamics of  $H_i$ as:

\begin{equation}
\frac{dH_i}{dt} = H_if(\textbf{X}),
\end{equation}
where  $f(\textbf{X}^*)$ is a function of all the densities. Then, for any density $X\neq H_i$,

\begin{equation}
\frac{\partial\dot{H}_i}{\partial X}\bigg|_{\textbf{X}^*} = \left( H_i\frac{\partial f(\textbf{X})}{\partial X}\right)\bigg|_{\textbf{X}^*}=0.
\end{equation}

\noindent Hence, the Jacobian  will have the form:

\begin{equation}
   J= \begin{pmatrix}
       \frac{\partial \dot{H_1}}{\partial H_1}  & \dots & \frac{\partial \dot{H_1}}{\partial H_n} & \frac{\partial\dot{H_1}}{\partial V_1} & \dots&\frac{\partial \dot{H_1}}{\partial V_n} \\[0.3em]
                   \vdots&   & & & & \vdots\\[0.3em]
       0\quad\dots&0\quad\frac{\partial \dot{H_i}}{\partial H_i}\quad 0 &\dots\quad0 &0&\dots&0\\[0.3em]
         \vdots&   & & & & \vdots\\[0.3em]
         \frac{\partial \dot{V_n}}{\partial H_1}  & \dots & \frac{\partial \dot{V_n}}{\partial H_n} & \frac{\partial\dot{V_n}}{\partial V_1} & \dots&\frac{\partial \dot{V_n}}{\partial V_n} \\[0.3em]
     \end{pmatrix}_{\textbf{X}^*}.
\end{equation}

\noindent Note that all the elements in the $i$th row are zero, except for the diagonal term, which has the form $\frac{\partial \dot{H_i}}{\partial H_i}\big|_{\textbf{X}^*}$. Thus, the characteristic equation can be written as:

\begin{equation}\label{eq:chaequation}
det(J-I\lambda) = \left(\frac{\partial \dot{H_i}}{\partial H_i}\bigg|_{\textbf{X}^*} -\lambda\right)P(\lambda)=0,
\end{equation}

\noindent where $P(\lambda)$ is a polynomial $\lambda$. Hence, $\frac{\partial \dot{H_i}}{\partial H_i}\bigg|_{\textbf{X}^*}$ is an eigenvalue of the Jacobian.

b)  The proof for the case where $V_i^*=0$ follows the same steps as  proof a) for $H_i^*=0$.
\end{proof}

\begin{lem}
\label{l:allVzeros}
Let  $\textbf{X}^*$ be a fixed point where  $V_i^*=0$ for all $i$ and $H_i^* \ne 0$ for at least one $i$, then $\textbf{X}^*$ is unstable with respect to invasion by at least one virus.
\end{lem}

\begin{proof}
Assume $H_i^*\neq0$. The equilibrium condition for $H_i$ is:

\begin{equation}
r_{i}\left(1 - \frac{\sum H_j^*}{K }\right)  - \cancel{\sum_{j=i}^n\phi_jV_j^*} = 0.
\end{equation}

\noindent So,  $ \sum H_j^* = K $.

Now, $V_n^*=0$, so using lemma \ref{l:eigenvalue} we know that one eigenvalue of the Jacobian is:

\begin{equation}
\label{eq:eigl2}
\lambda = \frac{\partial V_n^*}{\partial V_n}\bigg|_{X^*} = \phi_n\beta_n\sum H_j^* - m_n = \phi_n \beta_n K - m_n > 0,
\end{equation}

\noindent where the inequality results from the condition of coexistence \eqref{cond_K}. Equation \eqref{eq:eigl2} implies that the fixed point $\textbf{X}^*$ is unstable with respect to invasion by virus $n$.
\end{proof}

Now, we will show that fixed points with any of these 4 properties are unstable.

\begin{lem}
\label{l:allzeros}
Let  $\textbf{X}^*$ be a fixed point where $H_i^*=0$ for all $i\in [1,n]$, then $\textbf{X}^*$  is unstable with respect to invasion of any of the hosts.
\end{lem}

\begin{proof}
Assume $H_i^*=0$ for all types of hosts in the system, then from equation \eqref{eq:odeSumsH} it follows that also $V_i^*=0$ for all types of viruses at equilibrium.

Using  lemma \ref{l:eigenvalue},

\begin{equation}
\lambda = \frac{\partial\dot{H}_i}{\partial H_i}\bigg|_{X^*}=  r_{i}>0
\end{equation}

\noindent is an eigenvalue. Thus $\textbf{X}^*$ is unstable with respect to invasion of any of the hosts.
\end{proof}

\begin{lem}
\label{l:firstzeros}
Let  $\textbf{X}^*$ be a fixed point where  $H_1^*=H_2^*=H_3^*=\dots=H_k^*=0$ with $k<n$ and $H_{k+1}^*\ne0$, then $\textbf{X}^*$  is unstable with respect to invasion of hosts 1 through $k$ .
\end{lem}

\begin{proof}
Assume $H_1^*=H_2^*=H_3^*=\dots=H_k^*=0$ with $k<n$. The equilibrium condition for $V_i$ for  $1\leq i \leq k$ is:

\begin{align}
\left(\phi_i\beta_i\sum_{j=1}^i H_j - m_i\right)V_i = 0.
\end{align}

\noindent Note that $V_1,V_2,...,V_k$ only infect hosts 1 through $k$ (the upper limit of the sum is $i$ and $1\leq i \leq k$  ). Thus if  $H_i^*=0$  for the first $k$ viruses, then it must be the case that $V_i^*=0$  for the first $k$ viruses, i.e

\begin{equation}
\label{eq:firstVzeros}
V_1^*=V_2^*=V_3^*=\dots=V_k^*=0.
\end{equation}

\noindent From lemma \ref{l:eigenvalue} the eigenvalues for the first $k$ hosts have the form :

\begin{equation}\label{eq:eigcase3}
\lambda = \frac{\partial\dot{H}_i}{\partial H_i}\bigg|_{X^*}=  r_{i}\left(1 - \frac{\sum H_j^*}{K }\right)  - \sum_{j=k+1}^n\phi_jV_j^*.
\end{equation}

\noindent Note the lower limit in the sum as a result of equation \eqref{eq:firstVzeros}.

Now, assume  $H_{k+1}\neq 0 $. From the equilibrium condition for $H_{k+1}$ we get:

\begin{equation}\label{eq:case3}
\sum_{j=k+1}^n\phi_jV_j^* = r_{k+1}\left(1 - \frac{\sum H_j^*}{K }\right).
\end{equation}

\noindent If $V_i^*=0$ for all $i$ in the sum, then $V_i^*=0$ for all the viruses in the system. This case is already covered in lemma \ref{l:allVzeros}.

On the other hand if $V_i^*\neq0$ for at least one virus $V_i^*$ we get,

\begin{equation}\label{eq:condPos}
\left(1 - \frac{\sum H_j^*}{K}\right)>0.
\end{equation}

\noindent Plugging the result \eqref{eq:case3} in  equation \eqref{eq:eigcase3}, we get:

\begin{equation}\label{eq:eigl4}
\lambda = (r_i- r_{k+1})\left(1 - \frac{\sum H_j^*}{K }\right)>0,
\end{equation}

\noindent where the inequality comes from the inequality \eqref{eq:condPos} and  the condition for coexistence \eqref{cond_r}. Equation \eqref{eq:eigl4} implies that the fixed point $\textbf{X}^*$ is unstable with respect to invasion by host $i$ where $1\leq i\leq k$.

\end{proof}

\begin{lem}
\label{l:lastzeros}
Let  $\textbf{X}^*$ be a fixed point where $H_i^* = H_{i+1}^*=\dots=H_n^* = 0$ for $i\neq1$ and $H_{i-1}^*\neq0$, then $\textbf{X}^*$  is unstable with respect to invasion of at least one species of host or virus.
\end{lem}

\begin{proof}
Assume $H_i^* = H_{i+1}^*=\dots=H_n^* = 0$. The equilibrium  condition for  viruses $i-1$ through $n$  are:

\begin{equation}
\label{eq:equiVlast}
\begin{aligned}
\left( \sum^{i-1}_{j=1}H^*_j -  \frac{m_{i-1}}{\phi_{i-1}\beta_{i-1}}\right)V_{i-1}^* &= 0\\
\left( \sum^{i-1}_{j=1}H^*_j -  \frac{m_{i}}{\phi_{i}\beta_{i}}\right)V_{i}^* &= 0\\
\vdots\\
\left( \sum^{i-1}_{j=1}H^*_j -  \frac{m_{n}}{\phi_{n}\beta_{n}}\right)V_{n}^* &= 0.
\end{aligned}
\end{equation}

\noindent Note that all the sums in \eqref{eq:equiVlast} are the same. Because $\frac{m_\ell}{\phi_\ell\beta_\ell}\neq\frac{m_p}{\phi_p\beta_p}$ for all $\ell\neq p$ in the system, the equations in \eqref{eq:equiVlast} can only be satisfied if $V_k^*\neq0$ for at most one $k\in[i-1,n]$.

To see that $\textbf{X}^*$ is unstable, consider three possible cases:

\begin{enumerate}[{Case} 1:]
    \item  $\textbf{X}^*$  where $V_k^*=0 \quad\forall\quad k \in [i-1,n]$. \label{case:last1}
    \item  $\textbf{X}^*$  where $V_n^*\neq0$ and  $V_k^*=0 \quad\forall\quad k \in [i-1,n-1]$. \label{case:last2}
    \item  $\textbf{X}^*$  where $V_k^*\neq0$ for one $k \in [i-1,n-1]$ and $V_\ell=0$ for  $\ell \in [i-1,n]$, $k \neq \ell$. \label{case:last3}
\end{enumerate}

Case \ref{case:last1}:

Assume $V_k^*=0$ for all $k \in [i-1,n]$. In this case $V_k^*=0$ for all $k \in [1,n]$. Indeed, the equilibrium condition for $H_{i-1}$ implies:

\begin{equation}
 r_{i-1}\left(1 - \frac{\sum H_j^*}{K }\right) -\cancel{\sum_{j=i-1}^n\phi_jV_j^*} =0.
\end{equation}

\noindent Thus,

\begin{equation}
\label{eq:condllast}
 \sum H_j^* =K.
\end{equation}

\noindent From equation \eqref{eq:odeSumsH} we see that equation \eqref{eq:condllast} implies $V_i^*=0$ for all $i$. This is  the case covered in lemma \eqref{l:allVzeros}.

Case \ref{case:last2}:

Assume $V_n^*\neq0$ and  $V_k^*=0$ for all  $k \in [i-1,n-1]$. From lemma 1 we know that, because $V_k=0$ for $i-1\leq k<n$ the corresponding eigenvalues have the form:

\begin{equation}\label{eq:eiglast}
\lambda = \frac{\partial V_k^*}{\partial V_k}\bigg|_{X^*} = \phi_k\beta_k\sum_{j=1}^{i-1} H_j^* - m_k.
\end{equation}

\noindent From the equilibrium condition for $V_n$, we get:

\begin{equation}\label{eq:Vequilast}
 \sum_{j=1}^{i-1} H_j^* = \frac{m_n}{\phi_n\beta_n}.
\end{equation}

\noindent Note that the limits of the sums in \eqref{eq:eiglast} and \eqref{eq:Vequilast} are the same as a result of the $H_i^*$ that are zero, substituting the last  expression into  \eqref{eq:eiglast}:

\begin{equation}
\lambda = \phi_{k}\beta_{k}\frac{m_n}{\phi_n\beta_n} - m_{k} >0,
\end{equation}

\noindent where the inequality results from \eqref{cond_h}, and consequently the  steady state $\textbf{X}^*$ is unstable with respect to invasion of any of the viruses $i-1$ through $n-1$.

Case \ref{case:last3}:

$V_k^*\neq0$ for one $k \in [i-1,n-1]$ and $V_\ell=0$ for  $\ell \in [i-1,n]$, $k \neq \ell$. The only virus capable of infecting host $n$  is virus $n$.  Using lemma 1 we get  that one eigenvalue is given by:

\begin{equation}\label{eq:eiglemma5case2}
\lambda = \frac{\partial\dot{H}_n}{\partial H_n}\bigg|_{X^*}=r_{n}\left(1 - \frac{\sum H_j^*}{K }\right),
\end{equation}

\noindent where we used  $V_n^*=0$. At this fixed point the only virus infecting host $i-1$  is  virus $k$, so the equilibrium condition becomes:

\begin{equation}
r_{i-1}\left(1 - \frac{\sum H_j^*}{K }\right) -\phi_kV_k^* = 0.
\end{equation}

\noindent Thus,

\begin{equation}
\left(1 - \frac{\sum H_j^*}{K }\right)>0
\end{equation}

\noindent and

\begin{equation}
\lambda = r_{n}\left(1 - \frac{\sum H_j^*}{K }\right)>0.
\end{equation}

\noindent So, the fixed point $\textbf{X}^*$   is unstable with respect to invasion by host $n$.
\end{proof}

\begin{lem}
\label{l:middlezeros}
Let  $\textbf{X}^*$ be a fixed point where  $H_i^* = H_{i+1}^*=\dots=H_k^* = 0$ for $i\neq1$, $k<n$, and $H_{k+1}\neq0$, then $\textbf{X}^*$  is unstable with respect to invasion of at least one species of host or virus.
\end{lem}

Proof:

Similarly to lemma \ref{l:lastzeros} the equilibrium condition for  viruses $i-1$ through $k$ are:

\begin{align}\label{eq:lemma6}
\left( \sum^{i-1}_{j=1}H^*_j -  \frac{m_{i-1}}{\phi_{i-1}\beta_{i-1}}\right)V_{i-1}^* &= 0\\\nonumber
\left( \sum^{i-1}_{j=1}H^*_j -  \frac{m_{i}}{\phi_{i}\beta_{i}}\right)V_{i}^* &= 0\\\nonumber
\vdots\\\nonumber
\left( \sum^{i-1}_{j=1}H^*_j -  \frac{m_{k}}{\phi_{k}\beta_{k}}\right)V_{k}^* &= 0.
\end{align}

To simultaneously satisfy all of the equations at most one $V_\ell^* \neq 0$ for $\ell \in [i-1,k]$ . To show that $\textbf{X}^*$  is unstable,  it is convenient to  study  two different cases.

\begin{enumerate}[{Case} 1:]
    \item  $\textbf{X}^*$  where $V_k^*=0 $ \label{case:middle1}
    \item  $\textbf{X}^*$  where $V_k^*\neq0$  \label{case:middle2}.
\end{enumerate}

Case \ref{case:middle1}:

Assume $V_k=0$. Using lemma 1 we know that an eigenvalue of the Jacobian is:

\begin{equation}\label{eq:lemma6eig}
\lambda=\frac{\dot{H_k}}{H_k}= r_{k}\left(1 - \frac{\sum H_j^*}{K }\right) - \sum_{j=k+1}^n\phi_jV_j^*.
\end{equation}
The equilibrium condition for $H_{k+1}$ yields:

\begin{equation}
r_{k+1}\left(1 - \frac{\sum H_j^*}{K }\right) - \sum_{j=k+1}^n\phi_jV_j^*=0.
\end{equation}
Using this result the eigenvalue becomes:

\begin{equation}
\lambda = (r_{k} - r_{k+1})\left(1 - \frac{\sum H_j^*}{K }\right).
\end{equation}
And as a consequence of condition \eqref{cond_h}, the eigenvalue is positive and the fixed point is unstable with respect to invasion by host $k$.

Case \ref{case:middle2}:

Assume $V_k\neq0$. The equilibrium condition for $V_k$ is:

\begin{equation}\label{eq:case2eig}
\phi_k\beta_k\sum_{j=1}^{i-1}H_j^* - m_k =0.
\end{equation}

\noindent So,

\begin{equation}
\sum_{j=1}^{i-1}H_j^*=\frac{m_k}{\phi_k\beta_k}.
\end{equation}

\noindent From  Lemma 1  follows that for any $i \leq \ell<k$, there exists an eigenvalue of the form:

\begin{equation}
\lambda = \frac{\partial V_\ell^*}{\partial V_\ell}\bigg|_{X^*} = \phi_\ell\beta_\ell\sum_{j=1}^{i-1} H_j^* - m_\ell.
\end{equation}

\noindent Substituting equation \eqref{eq:case2eig}, we get:
\begin{equation}
\lambda =  \frac{\partial V_\ell^*}{\partial V_\ell}\bigg|_{X^*} = \phi_\ell\beta_\ell\frac{m_k}{\phi_k\beta_k}- m_\ell > 0,
\end{equation}

\noindent where the inequality is a consequence of condition \eqref{cond_h}. Thus, the fixed point  is unstable with respect to invasion by  viruses $i-1$ through $k-1$ .

\newpage

\section{Parameters values used in the numerical simulations}
\label{ap:para}

\begin{table}[h]
\begin{center}
\scriptsize
\begin{tabular}{c|cccccc}
\hline
     &      &  r  (h$^{-1}$)     &  K (cells/ml)     &  m (h$^{-1}$)     &  $\phi$ (ml/(cells$\times$h))&  $\beta$ (No. of viruses)\\ \hline
Fig.\ref{2h2vStable}&$H_1$ &  0.1760  & $10^7$ &     &          &         \\ \hline
                    &$H_2$ &  0.9914  & $10^7$ &     &          &         \\ \hline
                    &$V_1$ &          &        & 0.0681 &  $5\times10^{-9}$ & 24 \\ \hline
                    &$V_2$ &          &        & 0.1047 &  $5\times10^{-9}$ & 10  \\  \hline
  \hline
Fig.\ref{2h2vDie}&$H_1$ &   1.0500 & $10^8$ &     &          &         \\ \hline
                 &$H_2$ &   0.300 & $10^8$ &     &          &         \\ \hline
                 &$V_1$ &          &        & 0.0100 &  $5\times10^{-9}$ & 100 \\ \hline
                 &$V_2$ &          &        & 0.0650 &  $5\times10^{-9}$ & 100 \\ \hline
\hline
Fig.\ref{2h2vOsc}&$H_1$ &   0.3000 & $10^8$ &     &          &         \\ \hline
                &$H_2$ &   1.0500 & $10^8$ &     &          &         \\ \hline
                &$V_1$ &          &        & 0.0100 &  $5\times10^{-9}$ & 100 \\ \hline
                &$V_2$ &          &        & 0.0650 &  $5\times10^{-9}$ & 100 \\  \hline
\hline
Fig.\ref{5h5v}&$H_1$ &   0.3162 & $10^7$ &     &          &         \\ \hline
     &$H_2$ &   0.6325 & $10^7$ &     &          &         \\ \hline
     &$H_3$ &   0.8367 & $10^7$ &     &          &         \\ \hline
     &$H_4$ &   1.1000 & $10^7$ &     &          &         \\ \hline
     &$H_5$ &   1.1400 & $10^7$ &     &          &         \\ \hline
     &$V_1$ &          &        & 0.0090 &  $5\times10^{-9}$ & 20 \\ \hline
     &$V_2$ &          &        & 0.0270 &  $5\times10^{-9}$ & 40 \\ \hline
     &$V_3$ &          &        & 0.0548 &  $5\times10^{-9}$ & 60 \\ \hline
     &$V_4$ &          &        & 0.0922 &  $5\times10^{-9}$ & 80 \\ \hline
     &$V_5$ &          &        & 0.1392 &  $5\times10^{-9}$ & 100 \\ \hline
\hline
Fig.\ref{4by4}b&$H_1$ &   0.4472& $10^7$ &     &          &         \\ \hline
     &$H_2$ &   0.7583 & $10^7$ &     &          &         \\ \hline
     &$H_3$ &   0.9747 & $10^7$ &     &          &         \\ \hline
     &$H_4$ &   1.1511 & $10^7$ &     &          &         \\ \hline
     &$V_1$ &          &        & 0.0100&  $5\times10^{-9}$  & 20 \\ \hline
     &$V_2$ &          &        & 0.0375 &  $5\times10^{-9}$ & 46 \\ \hline
     &$V_3$ &          &        & 0.0650 &  $5\times10^{-9}$ & 70 \\ \hline
     &$V_4$ &          &        & 0.0925 &  $5\times10^{-9}$ & 96 \\ \hline
     \hline
\end{tabular}
\end{center}
\caption{Values of the life-history traits used in the numerical simulations. Values based on \cite{Wommack2000,Abedon2007}}
\label{table:para}
\end{table}

\newpage

\bibliographystyle{elsarticle-num}

\begin{thebibliography}{10}
\expandafter\ifx\csname url\endcsname\relax
  \def\url#1{\texttt{#1}}\fi
\expandafter\ifx\csname urlprefix\endcsname\relax\def\urlprefix{URL }\fi
\expandafter\ifx\csname href\endcsname\relax
  \def\href#1#2{#2} \def\path#1{#1}\fi

\bibitem{Suttle2005}
C.~A. Suttle, Viruses in the sea, Nature 437 (2005) 356--361.

\bibitem{Suttle2007}
C.~A. Suttle, Marine viruses: major players in the global ecosystem, Nat. Rev.
  Microbiol. 5 (2007) 801--812.

\bibitem{Suttle1994}
C.~A. Suttle, The significance of viruses to mortality in aquatic microbial
  communities, Microb. Ecol. 28 (1994) 237--243.

\bibitem{Weinbauer2004}
M.~G. Weinbauer, Ecology of prokaryotic viruses, FEMS Microbiol. Rev. 28 (2004)
  127--181.

\bibitem{Corinaldesi2010}
C.~Corinaldesi, A.~Dell'Anno, M.~Magagnini, R.~Danovaro, Viral decay and viral
  production rates in continental-shelf and deep-sea sediments of the
  mediterranean sea, FEMS Microbiol. Ecol. 72 (2010) 208--218.

\bibitem{Wichels1998}
A.~Wichels, S.~S. Biel, H.~R. Gelderblom, T.~Brinkhoff, G.~Muyzer,
  C.~Sch\"{u}tt, {Bacteriophage diversity in the North Sea}, Appl. Environ.
  Microbiol. 64 (1998) 4128--4133.

\bibitem{Holmfeldt2007}
K.~Holmfeldt, M.~Middelboe, O.~Nybroe, L.~Riemann, {Large variabilities in host
  strain susceptibility and phage host range govern interactions between lytic
  marine phages and their \textit{Flavobacterium} hosts}, Appl. Environ.
  Microbiol. 73 (2007) 6730--6739.

\bibitem{Gomez2011}
P.~G\'{o}mez, A.~Buckling, {Bacteria-phage antagonistic coevolution in soil},
  Science 332 (2011) 106--109.

\bibitem{Poisot2011}
T.~Poisot, G.~Lepennetier, E.~Martinez, J.~Ramsayer, M.~E. Hochberg, {Resource
  availability affects the structure of a natural bacteria-bacteriophage
  community}, Biol. Lett. 7 (2011) 201--204.

\bibitem{Weitz2013}
J.~S. Weitz, T.~Poisot, J.~R. Meyer, C.~O. Flores, S.~Valverde, M.~B. Sullivan,
  M.~E. Hochberg, Phage--bacteria infection networks, Trends Microbiol. 21
  (2013) 82--91.

\bibitem{Poullain2008}
V.~Poullain, S.~Gandon, M.~A. Brockhurst, A.~Buckling, M.~E. Hochberg, The
  evolution of specificity in evolving and coevolving antagonistic interactions
  between a bacteria and its phage, Evolution 62 (2008) 1--11.

\bibitem{Stenholm2008}
A.~R. Stenholm, I.~Dalsgaard, M.~Middelboe, {Isolation and characterization of
  bacteriophages infecting the fish pathogen \textit{Flavobacterium
  psychrophilum}}, Appl. Environ. Microbiol. 74 (2008) 4070--4078.

\bibitem{Sullivan2003}
M.~B. Sullivan, J.~B. Waterbury, S.~W. Chisholm, {Cyanophages infecting the
  oceanic cyanobacterium \emph{Prochlorococcus}}, Nature 424 (2003) 1047--1051.

\bibitem{Flores2011}
C.~O. Flores, J.~R. Meyer, S.~Valverde, L.~Farr, J.~S. Weitz, {Statistical
  structure of host-phage interactions}, {Proc. Natl. Acad. Sci. U. S. A.}
  {108} ({2011}) {E288--E297}.

\bibitem{Levin1977}
B.~R. Levin, F.~M. Stewart, L.~Chao, Resource-limited growth, competition, and
  predation: a model and experimental studies with bacteria and bacteriophage,
  Am. Nat. 111 (1977) 3--24.

\bibitem{Abedon2008}
S.~T. Abedon, Bacteriophage ecology: population growth, evolution, and impact
  of bacterial viruses, Cambridge University Press, Cambridge, UK, 2008.

\bibitem{Lenski1985}
R.~E. Lenski, B.~R. Levin, Constraints on the coevolution of bacteria and
  virulent phage: a model, some experiments, and predictions for natural
  communities, Am. Nat. (1985) 585--602.

\bibitem{Weitz2005}
J.~S. Weitz, H.~Hartman, S.~A. Levin, Coevolutionary arms races between
  bacteria and bacteriophage, Proc. Natl. Acad. Sci. U. S. A. 102 (2005) 9535.

\bibitem{Forde2008}
S.~E. Forde, R.~E. Beardmore, I.~Gudelj, S.~S. Arkin, J.~N. Thompson, L.~D.
  Hurst, Understanding the limits to generalizability of experimental
  evolutionary models, Nature 455 (2008) 220--223.

\bibitem{Weinberger2012}
A.~Weinberger, C.~Sun, M.~Pluci{\'n}ski, V.~Denef, B.~Thomas, P.~Horvath,
  R.~Barrangou, M.~Gilmore, W.~Getz, J.~Banfield, {Persisting viral sequences
  shape microbial CRISPR-based immunity}, PLoS Comput. Biol. 8 (2012) e1002475.

\bibitem{Childs2012}
L.~M. Childs, N.~L. Held, M.~J. Young, R.~J. Whitaker, J.~S. Weitz,
  {Multi-scale model of CRISPR-induced coevolutionary dynamics: diversification
  at the interface of Lamarck and Darwin}, Evolution 66 (2012) {2015--2029}.

\bibitem{Haerter2012}
J.~Haerter, K.~Sneppen, {Spatial structure and lamarckian adaptation explain
  extreme genetic diversity at CRISPR locus}, mBio 3 (2012) {e00126--12}.

\bibitem{Levin2010}
B.~R. Levin, {Nasty viruses, costly plasmids, population dynamics, and the
  conditions for establishing and maintaining CRISPR-mediated adaptive immunity
  in bacteria}, PLoS Genetics 6 (2010) e1001171.

\bibitem{Thingstad2000}
T.~F. Thingstad, {Elements of a theory for the mechanisms controlling
  abundance, diversity, and biogeochemical role of lytic bacterial viruses in
  aquatic systems}, {Limnol. Oceanogr.} {45} ({2000}) {1320--1328}.

\bibitem{Winter2010}
C.~Winter, T.~Bouvier, M.~G. Weinbauer, T.~F. Thingstad, Trade-offs between
  competition and defense specialists among unicellular planktonic organisms:
  the ``killing the winner'' hypothesis revisited, Microbiol. Mol. Biol. Rev.
  74 (2010) 42--57.

\bibitem{MATLAB2010}
MATLAB, version 7.11.584 (R2010b), The MathWorks Inc., Natick, Massachusetts,
  2010.

\bibitem{Wommack2000}
K.~E. Wommack, R.~R. Colwell, Virioplankton: viruses in aquatic ecosystems,
  Microbiol. Mol. Biol. Rev. 64 (2000) 69--114.

\bibitem{Abedon2007}
S.~T. Abedon, R.~R. Culler, Bacteriophage evolution given spatial constraint,
  J. Theor. Biol. 248 (2007) 111--119.

\bibitem{Hofbauer1998}
J.~Hofbauer, K.~Sigmund, Evolutionary games and population dynamics, Cambridge
  University Press, Cambridge, UK, 1998.

\bibitem{Lenski1988a}
R.~E. Lenski, {Experimental studies of pleiotropy and epistasis in
  \emph{Escherichia coli}. I. Variation in competitive fitness among mutants
  resistant to virus T4}, Evolution (1988) 425--432.

\bibitem{Bohannan2000}
B.~Bohannan, R.~Lenski, Linking genetic change to community evolution: insights
  from studies of bacteria and bacteriophage, Ecology Letters 3 (2002)
  362--377.

\bibitem{Breitbart2012}
M.~Breitbart, Marine viruses: truth or dare, Annu Rev Mar Sci 4 (2012)
  425--448.

\bibitem{Lennon2007}
J.~T. Lennon, S.~A.~M. Khatana, M.~F. Marston, J.~B. Martiny, Is there a cost
  of virus resistance in marine cyanobacteria?, ISME journal 1 (2007) 300--312.

\bibitem{Avrani2011}
S.~Avrani, O.~Wurtzel, I.~Sharon, R.~Sorek, D.~Lindell, {Genomic island
  variability facilitates \emph{Prochlorococcus}-virus coexistence}, Nature 474
  (2011) 604--608.

\bibitem{Duffy2006}
S.~Duffy, P.~E. Turner, C.~L. Burch, {Pleiotropic costs of niche expansion in
  the RNA bacteriophage $\Phi$6}, Genetics 172 (2006) 751--757.

\bibitem{Roswall2006}
M.~Rosvall, I.~B. Dodd, S.~Krishna, K.~Sneppen, Network models of
  phage-bacteria coevolution, Phys. Rev. E 74 (2006) 066105.

\bibitem{Rodriguez2009}
F.~Rodriguez-Valera, A.-B. Martin-Cuadrado, L.~P. Beltran Rodriguez-Brito,
  T.~F. Thingstad, A.~M. Forest~Rohwer, Explaining microbial population
  genomics through phage predation, Nat Rev Microbiol 7 (2009) 828--836.

\bibitem{Buckling2002}
A.~Buckling, P.~B. Rainey, Antagonistic coevolution between a bacterium and a
  bacteriophage, Proc. Natl. Acad. Sci. U. S. A. 269 (2002) 931--936.

\bibitem{Stern2011}
A.~Stern, R.~Sorek, The phage-host arms race: Shaping the evolution of
  microbes, Bioessays 33 (2011) 43--51.

\bibitem{Buckling2012}
A.~Buckling, M.~Brockhurst, Bacteria--virus coevolution, Evolutionary Systems
  Biology (2012) 347--370.

\bibitem{Kerr2006}
B.~Kerr, C.~Neuhauser, B.~J.~M. Bohannan, A.~M. Dean, Local migration promotes
  competitive restraint in a host--pathogen `tragedy of the commons', Nature
  442 (2006) 75--78.

\bibitem{thompson2005geographic}
J.~N. Thompson, {The geographic mosaic of coevolution}, University of Chicago
  Press, 2005.

\bibitem{Flores2012}
C.~O. Flores, S.~Valverde, J.~S. Weitz, Multi-scale structure and geographic
  drivers of cross-infection within marine bacteria and phages, {ISME} Journal
  (2012) 10.1038/ismej.2012.135.

\bibitem{green_2006}
J.~Green, B.~J. Bohannan, Spatial scaling of microbial biodiversity., Trends
  Ecol. Evol. 21 (2006) 501--507.

\bibitem{angly2006}
F.~E. Angly, B.~Felts, M.~Breitbart, P.~Salamon, R.~A. Edwards, C.~Carlson,
  A.~M. Chan, M.~Haynes, S.~Kelley, H.~Liu, J.~M. Mahaffy, J.~E. Mueller,
  J.~Nulton, R.~Olson, R.~Parsons, S.~Rayhawk, C.~A. Suttle, F.~Rohwer, The
  marine viromes of four oceanic regions, PLoS Biol 4~(11) (2006) e368.

\end{thebibliography}

\end{document}